\documentclass[conference]{IEEEtran}
\IEEEoverridecommandlockouts
\usepackage{cite}
\usepackage{amsmath,amssymb,amsfonts}
\usepackage{algorithmic}
\usepackage{graphicx}
\usepackage{textcomp}
\usepackage{xcolor}
\usepackage{amsthm}
\newtheorem{theorem}{Lemma}
\begin{document}

\title{Power Amplifier-Aware Transmit Power Optimization for OFDM and SC-FDMA Systems
\thanks{This research was funded by the Polish National Science Centre, project no. 2021/41/B/ST7/00136. For the purpose of Open Access, the author has applied a CC-BY public copyright license to any Author Accepted Manuscript (AAM) version arising from this submission.}
}

\author{\IEEEauthorblockN{Pawel KRYSZKIEWICZ}
\IEEEauthorblockA{\textit{Institute of Radiocommunications} \\
\textit{Poznan University of Technology}\\
Poznan, POLAND \\
pawel.kryszkiewicz@put.poznan.pl}
}

\maketitle

\begin{abstract}
The Single Carrier-Frequency Division Multiple Access (SC-FDMA) is a transmission technique used in the uplink of Long Term Evolution (LTE) and 5G systems, as it is characterized by reduced transmitted signal envelope fluctuations in comparison to Orthogonal Frequency Division Multiplexing (OFDM) technique used in the downlink. This allows for higher energy efficiency of User Equipments (UEs) while maintaining sufficient signal quality, measured by Error Vector Magnitude (EVM), at the transmitter. This paper proposes to model a nonlinear Power Amplifier (PA) influence while optimizing the transmit power in order to maximize the Signal to Noise and Distortion power Ratio (SNDR) at the receiver, removing the transmitter-based EVM constraint. An analytic model of SNDR for the OFDM system and a semi-analytical model for the SC-FDMA system are provided. Numerical investigations show that the proposed transmit power optimization allows for improved signal quality at the receiver for both OFDM and SC-FDMA systems. However, SC-FDMA still outperforms OFDM in this matter. Such a power amplifier-aware wireless transmitter optimization should be considered to boost the performance and sustainability of next-generation wireless systems, including Internet of Things (IoT) ones. 
\end{abstract}


\section{Introduction}
One of the important challenges for future wireless communication systems is to improve their sustainability. It is the most important for Internet of Things (IoT) devices that are typically battery-powered and require their energy efficiency to be high. One of the solutions, utilized, e.g., in User Equipments (UEs) of LTE or 5G New Radio (NR) systems, is Single Carrier-Frequency Division Multiple Access (SC-FDMA), called also Discrete Fourier Transform (DFT)-precoded Orthogonal Frequency Division Multiplexing (OFDM)\cite{Ochiai_SCFDMA_distribution_2012}. This allows to reduce complex envelope signal fluctuations, typically measured by Peak-to-Average Power Ratio (PAPR), in comparison to an OFDM signal. The higher the envelope fluctuations the higher the amount of nonlinear distortion while passing through a nonlinear Power Amplifier (PA) \cite{Gharaibeh_book_nonlinear}. Additionally, SC-FDMA can result in lower energy consumption by the PA \cite{Kryszkiewicz_Battery_2023}.

Contemporary, the nonlinear distortion is limited to comply with some performance metric at the transmitter, e.g., Error Vector Magnitude (EVM) \cite{3gpp_36101,3gpp_38141}. This can be achieved, e.g., by fixing the PA operating point, i.e., transmission power in relation to the clipping power of the PA. However, recently it has been observed for OFDM systems that it is possible to adjust the operating point in order to maximize the transmission quality \cite{Taveres_IBO_opt_OFDM_2016}, i.e., find balance between wanted signal power and nonlinear distortion power from the perspective of receiver-observed Signal to Noise and Distortion power Ratio (SNDR). The authors utilized a soft-limiter (clipper) PA model, that because of its simplicity, allowed for an analytic solution.   
Next, in \cite{Kryszkiewicz_PIMRC_2019} this concept has been extended to an OFDM system utilized in an edge computing scenario. The Rapp model, being a more generalized solution to a PA modeling, was used in \cite{Kryszkiewicz_Battery_2023}, though a closed-form expression was not obtained. Moreover, in addition to SNDR calculation energy efficiency maximization for OFDM was considered. The main conclusion of these papers is that dynamic operating point adjustment allows for higher throughput or even any transmission in a long-range link scenario.

However, similar considerations are missing for SC-FDMA systems, more suitable for uplink or IoT transmissions. For this purpose, SC-FDMA amplitude distribution is needed. Its approximation requiring though numerical integration is provided in \cite{Ochiai_SCFDMA_distribution_2012}. It has been used in \cite{Kun_SCFDMA_EVM_2017} to perform EVM analysis at the receiver for an SC-FDMA and OFDM transmission. To obtain results some data, e.g., SC-FDMA Power Spectral density (PSD) was obtained using simulations. As such the authors call their model semi-analytic. However, dynamic operating point adjustment has not been considered therein. 

This paper aims to derive the average SNDR formula for OFDM and SC-FDMA under soft-limiter PA to compare their performance under dynamic operating point adjustment. While OFDM allows for an analytic formula, a semi-analytic solution is provided for SC-FDMA with a theoretical justification of factors impacting SNDR. Moreover, while typically, time-domain SNDR calculation is used, e.g., in \cite{Taveres_IBO_opt_OFDM_2016, Kryszkiewicz_Battery_2023}, here a more accurate frequency-domain SNDR definition is provided, considering only in-band distortion. The numerical analysis is carried out showing that, similarly as in OFDM, dynamic operating point adjustment can significantly improve received signal quality for the SC-FDMA technique, though the optimal operating point depends on the constellation size and is different than for the OFDM. Moreover, the in-band-only distortion calculation is shown to result in a different operating point than the time-domain method and as such should be used.

The paper is organized as follows: The system model along with SNDR definitions for OFDM and SC-FDMA systems are shown in Sec. \ref{sec_System_model}. Next, the optimization problem for finding the SNDR maximizing operating point is shown in Sec. \ref{sec_SNDR_max}. The numerical evaluation results are shown in Sec. \ref{sec_simulation} that is followed by conclusions in Sec. \ref{sec_conclusions}.

\section{System Model}
\label{sec_System_model}
The considered transmitter is composed of an $N$-point Inverse Fast Fourier Transform (IFFT) block with the $n$-th sample of a single SC-FDMA/OFDM symbol on its output described as
\begin{equation}
    x_n=\sum_{k=0}^{N_{U}-1} d_k e^{j 2 \pi \frac{n I_k}{N}},
\end{equation}
where $d_k$ is the complex symbol modulating $I_k$-th subcarrier. The indices of all $N_{U}$ occupied subcarriers are denoted by a vector $I$ composed of unique, increasing (sorted in ascending order) elements of a total set of available subcarriers $\{ -N/2, ...., N/2-1\}$. In general, there is no assumption regarding the utilized subcarriers location, e.g., \cite{Ochiai_SCFDMA_distribution_2012} shown that the transmitted signal amplitude distribution, used in derivations in this paper, is invariant of the subcarriers allocation scheme. However, the performed simulations follow the Localized Frequency Division Multiple Access (LFDMA) approach, i.e., all utilized subcarriers are in a single block, similarly as in \cite{Kun_SCFDMA_EVM_2017}.
The unoccupied subcarriers are left to allow for digital-analog processing (reducing aliasing effect) or in order to reduce interference to other users. For OFDM the modulating symbols $d_k$ belong to one of the Quadrature Amplitude Modulation (QAM) or Phase Shift Keying (PSK) constellations.

However, for a SC-FDMA waveform the symbols $d_k$ are obtained as a result of $N_{U}$-point DFT processing as
\begin{equation}
    d_k=\frac{1}{\sqrt{N_{U}}}\sum_{\tilde{n}=0}^{N_{U}-1} \tilde{d}_{\tilde{n}} e^{-j 2 \pi \frac{\tilde{n} k}{N}},
\end{equation}
for $k\in \{0,...,N_{U}-1\}$. Here we assume $\tilde{d}_{\tilde{n}}$ are independent, complex symbols belonging to a chosen QAM/PSK constellation. 

Both for OFDM and SC-FDMA the transmitted time-domain signal has fixed mean power 
\begin{equation}
    \mathbb{E} \left[ |x_n|^2\right]=\sigma^2.
    \label{eq_power}
\end{equation}
Such a signal passes through a nonlinear PA of transfer function $\Gamma(~)$ giving output sample $y_n$ as
\begin{equation}
    y_n=\Gamma \left( x_n \right).
\end{equation}
While there are multiple PA models\cite{Gharaibeh_book_nonlinear}, here a soft-limiter will be used. First, it provides the highest, out of all possible nonlinearities, Signal to Distortion Ratio (SDR) in the case of OFDM signal \cite{Raich_optimal_nonlin_2005}. It is also possible that such a characteristic is obtained as an effective one for a PA preceded by a predistorter\cite{Wyglinski_predistortion_2016}. Moreover, in the case of a soft-limiter PA, analytical solutions for wanted signal and nonlinear distortion power exist \cite{Ochiai_2013_PA_efficiency,Taveres_IBO_opt_OFDM_2016}.
The PA output can be defined as
\begin{equation}
y_n=
\begin{cases}
x_n & \mbox{for}~~ |x_n|<\sqrt{P_{\mathrm{MAX}}} \\
\frac{\sqrt{P_{\mathrm{MAX}}}}{|x_n|}x_n & \mbox{for}~~ |x_n|\geq \sqrt{P_{\mathrm{MAX}}},
\end{cases}
\label{eq:PA_in_out}
\end{equation}
where $P_{\mathrm{MAX}}$ is the saturation power of the PA. An important parameter characterizing the operating point of the PA is Input Back-Off (IBO) specified as
\begin{equation}
    \gamma=\frac{P_{\mathrm{MAX}}}{\sigma^2}.
\end{equation}
Most importantly, the PA output signal can be decomposed into
\begin{equation}
y_n=\alpha x_n +q^{\mathrm{PA}}_{n},
\label{eq_busggang}
\end{equation}
where
\begin{equation}
    \alpha=\frac{\mathbb{E}\left[y_n x_n^* \right]}{\mathbb{E}\left[x_n x_n^* \right]}
    \label{eq_lambda_def}
\end{equation}
is a linear scaling factor of the input signal, and $q^{\mathrm{PA}}_{n}$ is a nonlinear distortion sample uncorrelated with $x_n$. While such a decomposition is commonly justified for OFDM signal, having complex Gaussian distribution\cite{Wei_2010_OFDM_dist}, by Bussgang theorem, it can be used for input signals of different characteristics, e.g., SC-FDMA, as well. While it was justified by the theory of homogenous linear mean square estimation in \cite{Kryszkieiwcz_ACTR_2018}, authors of \cite{Kun_SCFDMA_EVM_2017} confirmed its accuracy for SC-FDMA signal by means of simulations. 

Now, the problem is to obtain $\alpha$ and power of $q^{\mathrm{PA}}_{n}$. The first step can be the following Lemma:
 \begin{theorem}
 The linear scaling factor $\alpha$ for a soft-limiter PA is dependent on the IBO value, not the transmitted signal power itself for any input signal. The nonlinear distortion power is proportional to transmit signal power $\sigma^2$, with the proportionality coefficient dependent on IBO, not the transmit signal power $\sigma^2$. 
 \label{lemma1}
 \end{theorem}
  \begin{proof}

  It can be assumed that at the PA input samples $x_n$ are normalized by $\sigma$ creating 
  \begin{equation}
      \tilde{x}_n=\frac{x_n}{\sigma}
      \label{eq_x_norm}
  \end{equation}
  of unit mean power.
  By substituting it to (\ref{eq:PA_in_out}) a modified soft-limiter model is obtained as
  \begin{equation}
y_n=
\begin{cases}
\sigma\tilde{x}_n & \mbox{for}~~ |\tilde{x}_n|<\frac{\sqrt{P_{\mathrm{MAX}}}}{\sigma} \\
\frac{\sqrt{P_{\mathrm{MAX}}}}{|\tilde{x}_n|}\tilde{x}_n & \mbox{for}~~ |\tilde{x}_n|\geq \frac{\sqrt{P_{\mathrm{MAX}}}}{\sigma}.
\end{cases}
\label{eq:PA_in_out2}
\end{equation}
After substitution of $P_{\mathrm{MAX}}=\gamma \sigma^2$ it is obtained
  \begin{equation}
{y}_n=
\sigma
\underbrace{
\begin{cases}
\tilde{x}_n & \mbox{for}~~ |\tilde{x}_n|<\sqrt{\gamma} \\
\frac{\sqrt{\gamma}}{|\tilde{x}_n|}\tilde{x}_n & \mbox{for}~~ |\tilde{x}_n|\geq \sqrt{\gamma}.
\end{cases}}_{\tilde{y}_n=\tilde{\Gamma}\left(\tilde{x}_n\right)}
\label{eq:PA_in_out3}
\end{equation}
Observe that the nonlinear transfer function $\tilde{\Gamma}(~)$ maps the unit power signal $\tilde{x}_n$ to output signal $\tilde{y}_n$ and depends only on the IBO value. This can be used as equivalent nonlinearity to represent the mapping from $x_n$ to $y_n$ as 
\begin{equation}
    y_n=\sigma \tilde{\Gamma}\left(\frac{{x}_n}{\sigma}\right).
\end{equation}
After multiplying $\tilde{y}_n$ by $\sigma$ the signal $y_n$ is obtained. The decomposition in (\ref{eq_busggang}) can be defined for the equivalent nonlinearity $\tilde{\Gamma}(~)$ as
  \begin{equation}
\tilde{y}_n=\tilde{\alpha} \tilde{x}_n +\tilde{q}^{\mathrm{PA}}_{n},
\label{eq_busggang2}
\end{equation}
where $\tilde{\alpha}$ and $\tilde{q}^{\mathrm{PA}}_{n}$ are the new: scaling factor and nonlinear, uncorrelated distortion sample, respectively. The coefficient $\tilde{\alpha}$ depends on the IBO, not the transmit signal power $\sigma^2$ as a result of unit power input signal and function $\tilde{\Gamma}(~)$ parametrized only by IBO. For the same reason $\tilde{q}^{\mathrm{PA}}_{n}$ depends on IBO, not the transmit power itself.
By substituting ${y}_n=\sigma\tilde{y}_n$ and (\ref{eq_x_norm}) it is obtained
  \begin{equation}
{y}_n=\sigma\tilde{\alpha} \frac{{x}_n}{\sigma} +\sigma\tilde{q}^{\mathrm{PA}}_{n}=\tilde{\alpha}{x}_n +\sigma\tilde{q}^{\mathrm{PA}}_{n}
\label{eq_busggang3}
\end{equation}
By comparing it with (\ref{eq_busggang}) it is visible that both linear scaling coefficients are equal $(\alpha=\tilde{\alpha})$ and that the nonlinear signals are scaled by $\sigma$ as
\begin{equation}
    {q}^{\mathrm{PA}}_{n}=\sigma\tilde{q}^{\mathrm{PA}}_{n}.
    \label{eq_equivalent_dist}
\end{equation}
From (\ref{eq_equivalent_dist}) the nonlinear distortion power is
\begin{equation}
   \mathbb{E} \left[\left|{q}^{\mathrm{PA}}_{n}\right|^2\right]=\sigma^2 \mathbb{E} \left[\left|\tilde{q}^{\mathrm{PA}}_{n}\right|^2 \right],
    \label{eq_equivalent_dist2}
\end{equation}
where $\mathbb{E} \left[\left|\tilde{q}^{\mathrm{PA}}_{n}\right|^2 \right]$ can be treated as scaling factor dependent on IBO, not $\sigma^2$ as justified above (\ref{eq_busggang3}).
 \end{proof}

Considering (\ref{eq:PA_in_out}) and (\ref{eq_power}) the equation (\ref{eq_lambda_def}) simplifies to:
\begin{align}
\label{eq:alpha_by_pdf}
    \alpha&=\frac{\mathbb{E}\left[|y_n| |x_n|\right]}{\sigma^2}
    \\& \nonumber
    =\frac{1}{\sigma^2}\!
    \left(\int_{0}^{\sqrt{P_{\mathrm{MAX}}}}\!\!\!
   z^2 f_z(z) dz
   +\int_{\sqrt{P_{\mathrm{MAX}}}}^{\infty}
   \sqrt{P_{\mathrm{MAX}}}z f_z(z) dz
   \right),
   \end{align}
where $f_z(z)$ is Probability Density Function (PDF) of random variable $z$ being amplitude of PA input signal $x_n$. For the OFDM signal, the amplitude is known to be Rayleigh distributed, as a consequence of the complex-Gaussian distribution of samples $x_n$, which equals
\begin{equation}
\label{eq_Rayleigh}
    f_{\mathrm{z}}(z)=\frac{2z}{\sigma^2}e^{-\frac{z^2}{\sigma^2}},
\end{equation}
resulting in IBO-only dependent scaling factor \cite{Kryszkiewicz_Battery_2023}:
\begin{align}
\alpha^{\mathrm{OFDM}}(\gamma) &=
   1-e^{-\gamma}+\frac{1}{2}\sqrt{\pi \gamma} \mathrm{erfc}\left(\sqrt{\gamma}\right).
   \label{alpha_OFDM}
\end{align}
In (\ref{alpha_OFDM}) $\alpha$ was changed to $\alpha^{\mathrm{OFDM}}(\gamma)$ to emphasize dependence only on IBO value. 

Unfortunately, obtaining such a closed-form for the SC-FDMA signal is not possible. While \cite{Ochiai_SCFDMA_distribution_2012} provides an approximate Cumulative Distribution Function for instantaneous power of SC-FDMA signal, this requires numerical integration. Still, it can be used to obtain $\alpha$ using (\ref{eq:alpha_by_pdf}). Differently from the OFDM case, in SC-FDMA the constellation has a significant impact on amplitude distribution\cite{Ochiai_SCFDMA_distribution_2012}. As such $\alpha^{SC-FDMA}(\gamma,M)$ can be defined\footnote{\label{note1}While in SC-FDMA the whole constellation shape influences the results \cite{Ochiai_SCFDMA_distribution_2012}, here for brevity of notation this will be marked by the constellation size $M$.}. 

In order to calculate nonlinear distortion power, first the power of signal $y_n$ can be calculated using (\ref{eq_busggang}) minding the lack of correlation between $q^{\mathrm{PA}}_{n}$ and $x_n$ as
\begin{equation}
\mathbb{E}\left[ \left| y_n \right|^2 \right]=\left|\alpha\right|^2 \mathbb{E}\left[ \left|x_n\right|^2 \right] +\mathbb{E}\left[ \left|q^{\mathrm{PA}}_{n}\right|^2 \right].
\label{eq_busggang_pow}
\end{equation}
Therefore nonlinear distortion power equals 
\begin{equation}
\mathbb{E}\left[ \left|q^{\mathrm{PA}}_{n}\right|^2 \right]=\mathbb{E}\left[ \left| y_n \right|^2 \right]-\left|\alpha\right|^2 \sigma^2.
\label{eq_busggang_pow2}
\end{equation}
This needs the power of the PA output signal to be calculated as 
\begin{align}
\label{eq:y_power}
    \mathbb{E}\left[ \left| y_n \right|^2 \right]&=
    \int_{0}^{\sqrt{P_{\mathrm{MAX}}}}
   z^2 f_z(z) dz
   +\int_{\sqrt{P_{\mathrm{MAX}}}}^{\infty}
   P_{\mathrm{MAX}} f_z(z) dz.
   \end{align}
For the OFDM signal (\ref{eq_Rayleigh}) can be used resulting in \cite{Kryszkiewicz_Battery_2023}:
\begin{align}
\label{eq:y_power_OFDM}
    \mathbb{E}\left[ \left| y_n \right|^2 \right]&=\sigma^2 \left(1-e^{-\gamma} \right).
\end{align}
and the distortion power
\begin{equation}
\label{eq_dist_OFDM}
\mathbb{E}\left[ \left|q^{\mathrm{PA}}_{n}\right|^2 \right]=\sigma^2 D^{\mathrm{OFDM}}(\gamma)
\end{equation}
where
\begin{equation}
    D^{\mathrm{OFDM}}(\gamma)=1-e^{-\gamma}-|\alpha^{\mathrm{OFDM}}(\gamma)|^2.
\end{equation}
Observe, that the distortion power is proportional to the input signal power $\sigma^2$ by proportionality factor $D^{\mathrm{OFDM}}(\gamma)$ as shown by Lemma \ref{lemma1}. 

In the case of SC-FDMA (\ref{eq_busggang_pow2}) and (\ref{eq:y_power}) can be used, utilizing the numerically calculated CDF function from \cite{Ochiai_SCFDMA_distribution_2012}, resulting in  
\begin{equation}
\mathbb{E}\left[ \left|q^{\mathrm{PA}}_{n}\right|^2 \right]=\sigma^2 D^{\mathrm{SC-FDMA}}(\gamma,M)
\end{equation}
with the scaling coefficient dependent on IBO and QAM/PSK constellation order $M$\footref{note1}.

The signal $y_n$ passes through a wireless channel of $L$ taps and after the addition of noise sample $w_n$ can be presented as
\begin{align}
    z_n&=\sum_{l=0}^{L-1}h_l y_{n-l} +w_n
    \\&=\alpha\sum_{l=0}^{L-1}h_l x_{n-l}
   +\sum_{l=0}^{L-1}h_l q^{\mathrm{PA}}_{n-l}
    +w_n,
\end{align}
where $h_l$ is a complex channel coefficient for $l$-th tap. 
The mean received signal power can be calculated as
\begin{align}
\label{eq_RX_sig_power}
    \mathbb{E}\left[\left|z_n\right|^2\right]&=\left|\alpha\right|^2\sum_{l=0}^{L-1}\mathbb{E}\left[\left|h_l \right|^2\right]\mathbb{E}\left[\left|x_{n-l}\right|^2\right]
   \\& \nonumber +\sum_{l=0}^{L-1}\mathbb{E}\left[\left|h_l\right|^2\right] \mathbb{E}\left[\left|q^{\mathrm{PA}}_{n-l}\right|^2\right]
    +\mathbb{E}\left[\left|w_n\right|^2\right]
\end{align}
utilizing that the wanted signal, distortion, and noise signals are uncorrelated. Moreover, as the mean signal power and distortion power are invariant from the sample index a total channel gain can be defined as
\begin{equation}
    |\overline{h}|^2=\sum_{l=0}^{L-1}\mathbb{E}\left[\left|h_l \right|^2\right],
\end{equation}
and (\ref{eq_RX_sig_power}) can be rewritten as
\begin{align}
\label{eq_RX_sig_power2}
    \mathbb{E}\left[\left|z_n\right|^2\right]&=\left|\alpha\right|^2 |\overline{h}|^2 \mathbb{E}\left[\left|x_{n}\right|^2\right]+|\overline{h}|^2 \mathbb{E}\left[\left|q^{\mathrm{PA}}_{n-l}\right|^2\right]
    +\mathbb{E}\left[\left|w_n\right|^2\right],
\end{align}
with the first component being wanted signal power, the second the nonlinear distortion power, and the third the white noise power.
This allows to define Signal to Noise and Distortion power Ratio (SNDR) as 
\begin{equation}
    SNDR=\frac{\left|\overline{h}\right|^2\left|\alpha\right|^2\sigma^2}{\left|\overline{h}\right|^2\mathbb{E}\left[\left|q^{\mathrm{PA}}_{n-l}\right|^2\right]+\mathbb{E}\left[\left|w_n\right|^2\right]}.
\end{equation}
For OFDM this can be made more specific by using (\ref{eq_dist_OFDM}) and (\ref{alpha_OFDM}) as 
\begin{align}
    SNDR^{\mathrm{OFDM}}&=\frac{\left|\overline{h}\right|^2\left|\alpha^{\mathrm{OFDM}}(\gamma)\right|^2\sigma^2}{\left|\overline{h}\right|^2
    \sigma^2 D^{\mathrm{OFDM}}(\gamma)+\mathbb{E}\left[\left|w_n\right|^2\right]}
    \label{eq_SNDR_OFDM}
    \\& \nonumber
    =\frac{\left|\alpha^{\mathrm{OFDM}}(\gamma)\right|^2}{
    D^{\mathrm{OFDM}}(\gamma)
    +\frac{\gamma}{SNR_{\mathrm{SAT}}}
     },
\end{align}
where $SNR_{\mathrm{SAT}}= \frac{\left|\overline{h}\right|^2 P_{\mathrm{MAX}}}{\mathbb{E}\left[\left|w_n\right|^2\right]}$ is a Saturation SNR independent of the chosen operating point of the PA, used similarly in \cite{Kryszkiewicz_Battery_2023,Taveres_IBO_opt_OFDM_2016}. The $SNR_{\mathrm{SAT}}$ can be treated as an SNR if maximum power unmodulated carrier is transmitted.

However, the above expression characterizes signals in the time domain over the whole distorted signal bandwidth. In practice, a part of the nonlinear signal will leak into the out-of-band frequency region, as such not influencing the reception performance. In \cite{lee2014characterization} it was shown that around 2/3 of the distortion power falls within the in-band frequency range. As such the in-band SNDR is defined as
\begin{align}
    \tilde{SNDR}^{\mathrm{OFDM}}&=\frac{\left|\alpha^{\mathrm{OFDM}}(\gamma)\right|^2}{
    \tilde{D}^{\mathrm{OFDM}}(\gamma)
    +\frac{\gamma}{SNR_{\mathrm{SAT}}}
     },
\end{align}
where 
\begin{equation}
    \tilde{D}^{\mathrm{OFDM}}(\gamma)=\frac{2}{3}\left(1-e^{-\gamma}-|\alpha^{\mathrm{OFDM}}(\gamma)|^2\right).
\end{equation}
Taking the frequency-domain characteristic of the wanted signals into account the $SNR_{\mathrm{SAT}}$ should be specified with only noise power observed over utilized subcarriers band, i.e., $SNR_{\mathrm{SAT}}= \frac{\left|\overline{h}\right|^2 P_{\mathrm{MAX}}}{N_{\mathrm{U}}/N\mathbb{E}\left[\left|w_n\right|^2\right]}$. This, improved $SNR_{\mathrm{SAT}}$ definition should be used for both $SNDR$ and $\tilde{SNDR}$ calculation.

For the SC-FDMA the time-domain SNDR can be defined, following the approach from (\ref{eq_SNDR_OFDM}), as
\begin{align}
    SNDR^{\mathrm{SC-FDMA}}&=\frac{\left|\alpha^{\mathrm{SC-FDMA}}(\gamma,M)\right|^2}{
    D^{\mathrm{SC-FDMA}}(\gamma,M)
    +\frac{\gamma}{SNR_{\mathrm{SAT}}}
     }.
\end{align}
Similarly, the in-band SNDR can be defined as 
\begin{align}
    \tilde{SNDR}^{\mathrm{SC-FDMA}}&=\frac{\left|\alpha^{\mathrm{SC-FDMA}}(\gamma,M)\right|^2}{
    \tilde{D}^{\mathrm{SC-FDMA}}(\gamma,M)
    +\frac{\gamma}{SNR_{\mathrm{SAT}}}
     }.
\end{align}
However, similarly as in \cite{Kun_SCFDMA_EVM_2017} $\tilde{D}^{\mathrm{SC-FDMA}}(\gamma,M)$ has to be obtained by means of Monte Carlo simulations. 

\section{SNDR maximization by PA's operating point adjustment}
\label{sec_SNDR_max}
In state-of-the-art wireless systems, the nonlinear distortion is limited at the transmitter, e.g., by means of Error Vector Magnitude specified in standards \cite{3gpp_38141,3gpp_36101}. This is equivalent to providing a sufficiently high Signal to Distortion Ratio (SDR). This allows the transmitted signal to be quasi-distortionless simplifying the system optimization.

However, as shown in the previous section the SNDR, both in OFDM and SC-FDMA systems depend on the IBO value. As such the SNDR maximization at the receiver can be defined as:
\begin{equation}
    \max_{\gamma} SNDR
\end{equation}
where the SNDR can be either time-domain or in-band only and defined for OFDM or SC-FDMA systems.

The optimal IBO for OFDM can be found by calculating the first derivative of SNDR over IBO which is followed by finding the root using the Newton-Raphson method like in \cite{Taveres_IBO_opt_OFDM_2016}.
As for the SC-FDMA, the fully numerical solutions, e.g., gradient descent algorithm, are the only possible. 

\section{Numerical Results}
\label{sec_simulation}
Here an OFDM/SC-FDMA link is considered of $N=512$ and $N_{\mathrm{U}}=24$. The $N$ should be much larger than $N_{\mathrm{U}}$ to accommodate nonlinear distortion out-of-band emission. Additionally, $N_{\mathrm{U}}$ should be large enough to justify the complex-Gaussian distribution of $x_n$ samples\cite{Dinis_Gaussian_OFDM_2012}. As shown in \cite{Taveres_IBO_opt_OFDM_2016} the constellation order $M$ or IFFT size $N$ does not have a significant influence on the SNDR results for OFDM. However, as the utilized constellation has a substantial impact on SC-FDMA distribution \cite{Ochiai_SCFDMA_distribution_2012}, QPSK ($M=4$) and 64QAM ($M=64$) has been considered in this case.

\begin{figure}[htbp]
\centerline{\includegraphics[width=0.48\textwidth]{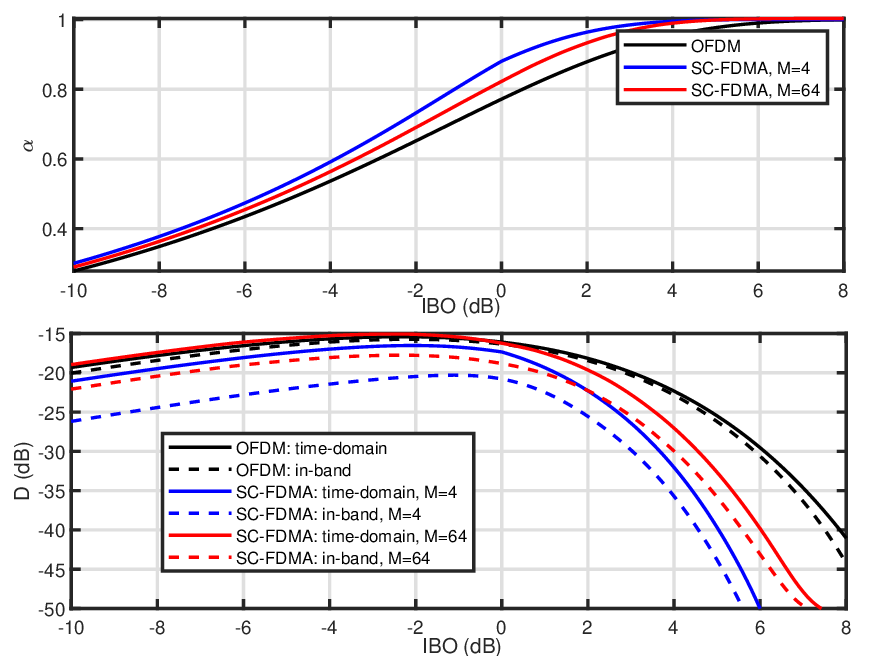}}
\caption{Coefficients $\alpha$ and $D$ as a function of IBO for OFDM/SC-FDMA (QPSK, 64QAM modulations) considering in-band or time-domain signals.}
\label{fig_alpha_D}
\end{figure}
First, values of $\alpha$ and $D$ have been evaluated for all considered systems while varying IBO value as visible in Fig. \ref{fig_alpha_D}. As expected $\alpha$ increases with IBO both for SC-FDMA and OFDM from 0 to 1. However, there is a difference in curves for all three cases. The wanted signal is the weakest for OFDM at a given IBO value. The distortion coefficient $D$ is not monotonic as a function of IBO. However, the curve shape is quite similar for all systems. As expected, the only in-band distortion is always weaker than the total (time-domain) distortion with the difference being the highest for the SC-FDMA system reaching up to 5 dB. 

\begin{figure}[htbp]
\centerline{\includegraphics[width=0.48\textwidth]{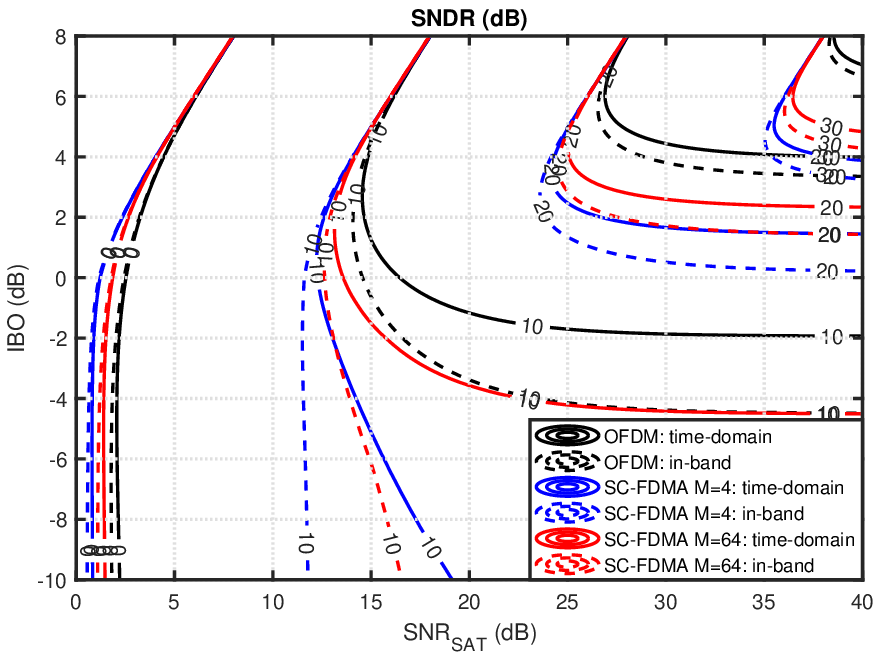}}
\caption{Contour plot of SNDR (in-band or time-domain) for SC-FDMA and OFDM signals as a function of IBO and $SNR_{\mathrm{SAT}}$.}
\label{fig_SNDR_contour}
\end{figure}
In Fig. \ref{fig_SNDR_contour} a contour plot of SNDR as a function of IBO and $SNR_{\mathrm{SAT}}$ is shown. It is visible that in general SNDR rises with the link quality (measured by $SNR_{\mathrm{SAT}}$) but there is an optimal IBO in each case. 

\begin{figure}[htbp]
\centerline{\includegraphics[width=0.48\textwidth]{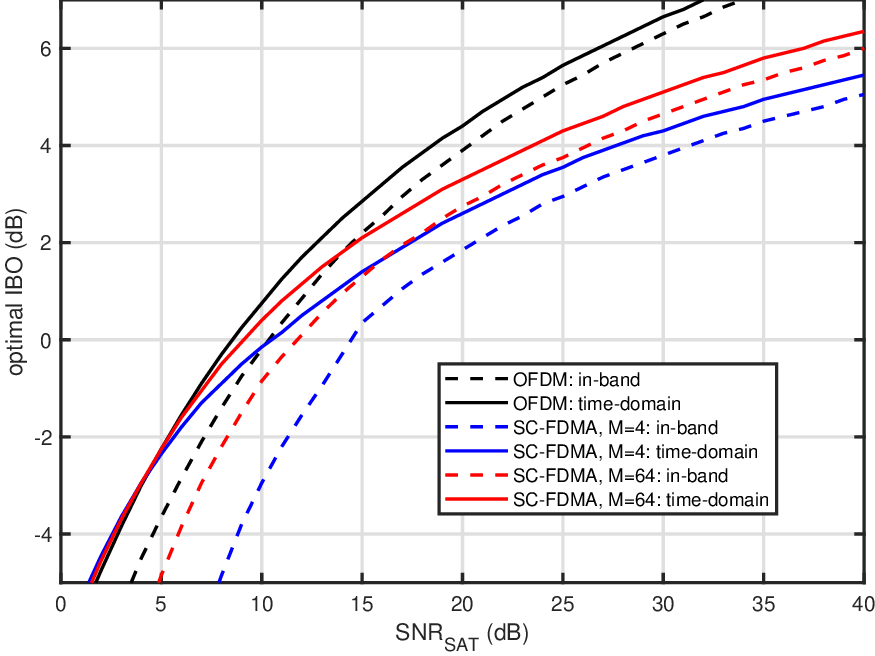}}
\caption{Optimal IBO as a function of $SNR_{\mathrm{SAT}}$ for SC-FDMA and OFDM signals, both in-band and time-domain distortion.}
\label{fig_optimal_IBO}
\end{figure}
This allows us to perform SNDR maximization for each $SNR_{\mathrm{SAT}}$ as explained in Sec. \ref{sec_SNDR_max}. The optimal IBO values are plotted as a function of $SNR_{\mathrm{SAT}}$ in Fig. \ref{fig_optimal_IBO}. It is visible that the optimal IBO increases with the link quality. When the link quality is poor (e.g.,  $SNR_{\mathrm{SAT}}<10 dB$) it is worth significantly clipping the signal, even blow mean signal power of the input signal, i.e., IBO below 0 dB. Most importantly, these observations are common for all considered systems, not only OFDM as observed in \cite{Taveres_IBO_opt_OFDM_2016,Kryszkiewicz_Battery_2023}. However, in general, the SC-FDMA requires lower IBO values than the OFDM system. Most significantly, if the modeling is carried according to the more accurate SNDR model, considering only in-band distortion, the optimal IBO is lower than in the previous, time-domain approach, sometimes by a few dB. As such it can be recommended to use only the improved SNDR definition. 

Finally, the maximal in-band SNDR, optimized considering in-band only distortions, as a function of $SNR_{\mathrm{SAT}}$ is shown in Fig. \ref{fig_max_SNDR}. First, it is visible that the maximal SNDR rises nearly linearly with $SNR_{\mathrm{SAT}}$ for every considered system, though the optimal SNDR differs among systems. Expectedly, the highest SNDR is achievable for SC-FDMA using QPSK modulation as this should have the smallest envelope fluctuations. However, the improvement with respect to the OFDM is not higher than 2 dB. As the optimal SNDR curves are nearly linear the least squared approximation was carried. The obtained linear functions, with coefficients shown in the legend of Fig. \ref{fig_max_SNDR}, constitute quite an accurate approximation of achievable SNDR while varying IBO value. In this figure additionally, reference solutions are shown based on EVM constraints from 5G/LTE standards. For OFDM a fixed IBO equal to 6 dB was used, as this results in SDR around 27 dB which is equivalent to 4.5\% EVM required for 256 QAM downlink transmission from 5G New Radio Base Station\cite{3gpp_38141}. With a similar reasoning based on \cite{3gpp_36101}, reference IBO was set to $-10$ dB and $2$ dB for SC-FDMA waveform utilizing QPSK and 64 QAM, respectively. 
It is visible that the reference solution only \emph{touches} the optimal curve for a $SNR_{\mathrm{SAT}}$ for which optimal IBO (from Fig. \ref{fig_optimal_IBO}) is equal to the reference scenario. As such for other $SNR_{\mathrm{SAT}}$ values some gain in SNDR can be achieved. For lower $SNR_{\mathrm{SAT}}$ values the SNDR gain is possible as a result of increased transmission power over a bad-quality channel. This increases wanted signal power at the cost of increased nonlinear distortion power. However, as the reception is white noise-limited, the increased distortion power is acceptable.  On the other hand, the SNDR gain for high $SNR_{\mathrm{SAT}}$ is obtained by decreasing nonlinear distortion power, being the main source of distortion at the receiver and, at the same time, decreasing the wanted signal power. 

\begin{figure}[htbp]
\centerline{\includegraphics[width=0.48\textwidth]{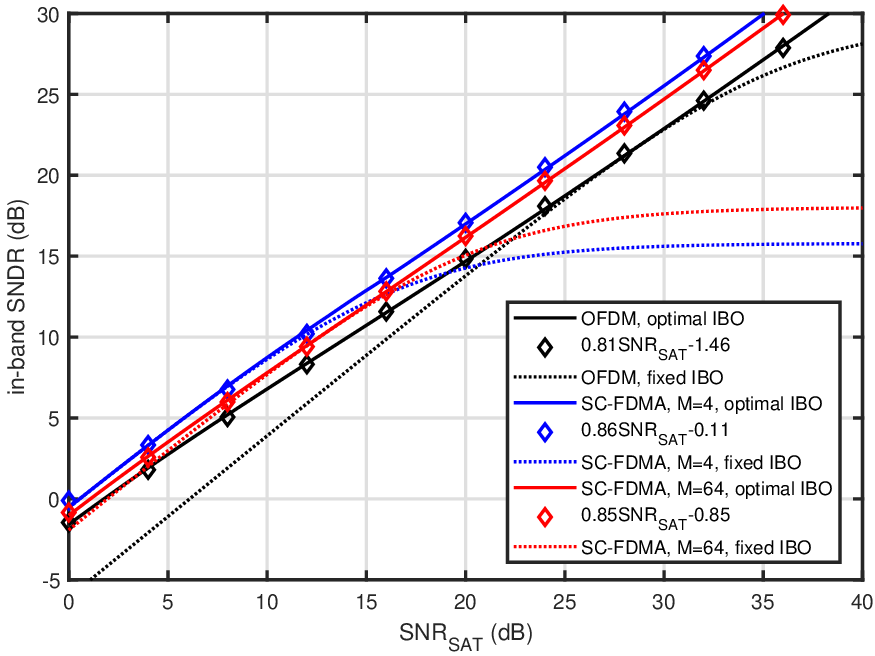}}
\caption{In-band SNDR for optimized IBO, with linear approximation, in comparison to a fixed IBO system as a function of $SNR_{\mathrm{SAT}}$ for SC-FDMA/OFDM.}
\label{fig_max_SNDR}
\end{figure}
\section{Conclusions}
\label{sec_conclusions}
The paper shows that adjustment of the PA operating point can boost significantly received signal quality with respect to the constant operating point. This is valid both for OFDM and SC-FDMA systems though the optimal operating point is different. Moreover, it has been shown that the in-band-only distortion model should be used instead of a time-domain model as it significantly influences the optimal solution. The next research direction can be to redefine the operating point adjustment goal, e.g., to maximize the energy efficiency of the transmission.
\bibliographystyle{IEEEtran}
\bibliography{biblio}

\end{document}